\documentclass[a4paper,twocolumn,11pt,unpublished]{quantumarticle}
\pdfoutput=1
\usepackage[T1]{fontenc}
\usepackage[utf8]{inputenc}
\usepackage[english]{babel}
\usepackage{microtype}
\usepackage{amsmath,amssymb,amsthm,mathtools,bm}
\usepackage{graphicx,booktabs,array,tabularx}
\usepackage[numbers,sort&compress]{natbib}
\usepackage[colorlinks=true,allcolors=quantumviolet]{hyperref}
\usepackage[nameinlink,noabbrev]{cleveref}
\makeatletter
\patchcmd{\@printtitle}{\@printtitletextwithappropriatefontsize}
  {\color{quantumviolet}\@printtitletextwithappropriatefontsize}{}
  {\PackageError{manuscript}{Unable to apply violet title styling}{}}
\makeatother
\hypersetup{pdftitle={Complete Detector Records and Contextual Source Laws in a Retrocausal Spin Model},pdfauthor={D.M. Theshan N. Weerasinghe}}
\usepackage{enumitem}
\setlist{nosep}
\newtheorem{theorem}{Theorem}[section]
\newtheorem{proposition}[theorem]{Proposition}
\newtheorem{corollary}[theorem]{Corollary}
\newtheorem{definition}[theorem]{Definition}
\newcommand{\Tr}{\operatorname{Tr}}
\newcommand{\sech}{\operatorname{sech}}
\newcommand{\TV}{\operatorname{TV}}
\newcommand{\dd}{\mathrm{d}}
\newcommand{\id}{\mathbb I}
\newcommand{\ket}[1]{\lvert #1\rangle}
\newcommand{\bra}[1]{\langle #1\rvert}

\newcommand{\E}{\mathbb E}
\numberwithin{equation}{section}
\title{Complete Detector Records and Contextual Source Laws in a Retrocausal Spin Model}
\author{D.M. Theshan N. Weerasinghe}
\email[\newline]{theshan.dissanayakemudiyanselage1@monash.edu}
\affiliation{Department of Electrical and Computer Systems Engineering, Faculty of Engineering, Monash University, Clayton, Australia}
\date{}
\begin{document}
\normalfont\selectfont
\maketitle
\begin{abstract}
Detector outcome probabilities can be independent of a later measurement setting even when earlier timestamps, detector seeds, or environmental records depend on that setting. We characterize complete-record preservation for a positive wrapped-Cauchy retrocausal spin source. An antipodal identity yields necessary and sufficient conditions for passive and environmental readouts, and source compensation permits a common normalized analyzer for source spins and prepared continuations. Gaussian controller and reset calculations distinguish inverse raw history weights from normalized thermal operations. A quantum instrument provides an operational benchmark, whereas a setting-dependent separable source produces setting-dependent statistics at an earlier joint probe. For two equal-width contexts, a calibrated weak probe requires a minimum source change of $\sech(2g)/8$ in trace distance for exact protection at every nonzero strength. The signal--source-change tradeoff is sharp for this pair of contexts. A contextual positive-history construction recovers the isolated hidden law and preserves records for a specified class of finite adaptive quantum apparatuses, subject to quantum calibration and restricted field access. Direct readout of an added likelihood field distinguishes two future axes with unit probability. The analytical results are supported by continuous-angle and quantum-trajectory simulations, deterministic quadrature, and exhaustive enumeration of a finite policy class. The analysis concerns the specified statistical models; their microscopic physical realization remains an open question.
\end{abstract}

\section{Introduction and scope}\label{sec:intro}
Retrocausal source descriptions permit variables associated with a preparation to depend on later measurement settings. This dependence need not appear in detector outcome marginals. In the anomalous-rotation model studied by Almada et al. \cite{almada}, antipodal symmetry protects the isolated bipartite outcome probabilities over a continuous range of propagation widths. More broadly, locally mediated formulations allow boundary data at different times to constrain a common history \cite{wharton,argaman}. The joint-outcome criterion of independence from later inputs is formulated in Ref.~\cite{argaman}; the present analysis specializes it to explicit source and apparatus extensions. An isolated source prescription leaves open which additional records can be acquired and archived before later controls are selected while maintaining a distribution independent of those controls.

A normalized readout can preserve the isolated outcome probabilities after marginalization while its own distribution carries information about a future setting. Likewise, leaving the value of an archive fixed along a history does not preserve its probability when a later intervention reweights histories. Record consistency therefore concerns the joint probability measure of all accessible earlier outputs. The accessible observables form part of the apparatus specification, since classical copies of hidden variables and quantum measurements of an environment need not define equivalent extensions.

This study establishes a source-specific antipodal characterization, distinguishes averaged detector kernels from their recorded realizations, analyzes controller normalization, derives a sharp calibrated-probe tradeoff, and constructs a conditional contextual embedding with an explicit obstruction to unrestricted field access. The analysis uses established results on competing rates, Gaussian partition functions, trace-preserving quantum instruments, and conditional reweighting \cite{maes,seifert,jarzynski,horodecki,stinespring}. The definitions, source law, detector constructions, proofs, and numerical procedures are presented below.

The results apply under distinct assumptions. Passive and environmental theorems assume the stated source law is unchanged when the readout is appended, and necessity uses a continuum of future axes. The calibrated-probe theorem allows the source to change, but assumes that an input density operator with fixed probe effects predicts the actual earlier probabilities. The contextual construction replaces that last assumption: its quantum calibration tensor determines records, while hidden source mixtures need not be operational input states. Its domain contains one preparation, one or two outgoing arms, local apparatus sequences, and causal classical feed-forward; it excludes recombination and arbitrary multiple-source networks. These assumptions specify the domain of each result.

\section{Source law and protection of complete records}\label{sec:records}
\subsection{Positive source weights and the record criterion}
Let $\lambda\in\mathbb T=\mathbb R/(2\pi\mathbb Z)$ be a source angle; Alice and Bob choose axes $x,y$, obtain outcomes $a,b\in\{+1,-1\}$, and have fixed finite widths $\gamma_A,\gamma_B>0$. Define
\begin{align}
 p_\gamma(t)&=\frac{\sinh\gamma}{2\pi(\cosh\gamma-\cos t)},\label{eq:kernel}\\
 \theta_{z,c}&=z+\pi\mathbf1_{c=-1},\\
 A_a(\lambda)&=p_{\gamma_A}(\theta_{x,a}-\lambda),\\
 B_{y,b}(\lambda)&=p_{\gamma_B}(\theta_{y,b}-\lambda-\pi),\\
 R_A(\lambda)&=\sum_a A_a(\lambda),\\
 R_B(\lambda;y)&=\sum_b B_{y,b}(\lambda),\\
 Z_{xy}&=\int_0^{2\pi}R_A(\lambda)R_B(\lambda;y)\dd\lambda.
\end{align}
The shift $\pi$ describes opposite source spins. Integrals use Lebesgue angle measure unless stated otherwise; a common uniform-source factor cancels in normalized probabilities. Both $R_A$ and $R_B$ are positive and $\pi$-periodic. The isolated hidden joint law is
\begin{equation}
 P_0(\dd\lambda,a,b\mid x,y)=\frac{A_a(\lambda)B_{y,b}(\lambda)}{Z_{xy}}\dd\lambda.\label{eq:original}
\end{equation}
The Fourier series $p_\gamma(t)=(2\pi)^{-1}\sum_{n\in\mathbb Z}e^{-|n|\gamma}e^{int}$ shows that convolution adds widths. With $\Gamma=\gamma_A+\gamma_B$,
\begin{align}
 \int A_aB_{y,b}\dd\lambda&=\frac{\sinh\Gamma}{2\pi}\frac1{\cosh\Gamma+ab\cos(x-y)},\\
 \widehat W^{ab}_{xy}&=\frac1{\cosh\Gamma+ab\cos(x-y)},\\
 \widehat Z_{xy}&=\frac{4\cosh\Gamma}{\cosh^2\Gamma-\cos^2(x-y)},\\
 Z_{xy}&=\frac{\sinh\Gamma}{2\pi}\widehat Z_{xy},\\
 P_0(a,b\mid x,y)&=\frac14[1-ab\sech\Gamma\cos(x-y)].\label{eq:pair}
\end{align}
The outcomes separately have probability $1/2$, even though the conditional hidden density
\begin{equation}
 \mu_{x,a,y}(\lambda)=\frac{2A_a(\lambda)R_B(\lambda;y)}{Z_{xy}}\label{eq:conditional}
\end{equation}
generally depends on $y$. For the equal-width source, the individual width $g$ here corresponds to the single-arm parameter of Ref.~\cite{almada}; its effective pair width is $2g$. Allowing unequal positive widths is the direct convolution generalization. The apparatus extensions considered below supplement this source law with specified readout and dynamical prescriptions.

\begin{definition}[Record consistency]
An earlier record $h$ is the tuple of all jointly accessible outputs retained at a specified earlier cut, possibly including an outcome, time, detector seed, and environmental or reset outputs. A family of later interventions preserves this record when its probability measure is independent of the later intervention. For a causal adaptive choice $u=f(h)$, the same earlier measure must result after summing all later outputs. The record includes all jointly accessible outputs of the specified apparatus and does not assign simultaneous values to incompatible quantum observables.
\end{definition}
For a finite or countable supported record distribution $p(h)>0$, raw continuation mass $L(h,u)>0$, and scalar correction $C(h,u)>0$, global normalization gives
\begin{equation}
 p'_u(h)=\frac{p(h)C(h,u)L(h,u)}{\sum_{h'}p(h')C(h',u)L(h',u)}.\label{eq:scalar}
\end{equation}
At fixed $u$, protection holds exactly when $CL$ is independent of $h$. If at least two branches have independently selectable controls and every such policy is admitted, comparison with a held-fixed second branch forces one common value across allowed branch-control pairs. The sufficient form is $C=L_{\rm ref}/L$ for a single positive reference. It protects branch mass; arbitrary additional readouts need separate analysis.

\subsection{Passive readouts, clocks, and adaptive settings}
Let $L_{x,a}(\dd r\mid\lambda)$ be a normalized Markov kernel on a standard Borel readout space, with no explicit dependence on $y$. ``Passive'' means that appending it otherwise leaves \cref{eq:original} unchanged. Its earlier law is
\begin{equation}
 P(\dd r\mid x,a,y)=\int\mu_{x,a,y}(\lambda)L_{x,a}(\dd r\mid\lambda)\dd\lambda.\label{eq:readout}
\end{equation}
Normalization only ensures recovery of the original law after $r$ is ignored.

\begin{theorem}[Antipodal characterization]\label{thm:passive}
Fix $x,a$ and positive finite widths. Write $A(\lambda)=A_a(\lambda)$ and $A_\pi(\lambda)=A(\lambda+\pi)$. The measure in \cref{eq:readout} is independent of every $y\in[0,\pi)$ if and only if there is a probability measure $\ell$ such that, almost everywhere in $\lambda$,
\begin{equation}
 \begin{aligned}
 &A(\lambda)L(\dd r\mid\lambda)\\
 &\quad+A_\pi(\lambda)L(\dd r\mid\lambda+\pi)\\
 &\qquad=[A(\lambda)+A_\pi(\lambda)]\ell(\dd r).
 \end{aligned}\label{eq:antipodal}
\end{equation}
\end{theorem}
\begin{proof}
For each readout event $E$, protection is equivalent to
the condition
\begin{equation}
 \int A(\lambda)[L(E\mid\lambda)-\ell(E)]R_B(\lambda;y)\dd\lambda=0
\end{equation}
for every translate $y$. The translates of $R_B$ have nonzero coefficients at every even Fourier harmonic, proportional to $e^{-2|n|\gamma_B}$. Thus all even coefficients of the integrable function $A[L(E\mid\lambda)-\ell(E)]$ vanish. Uniqueness of Fourier coefficients identifies this with vanishing $\pi$-periodization, giving \cref{eq:antipodal} for $E$. A countable determining class for probability measures on the standard Borel readout space gives one angle null set and the measure identity. Conversely, antipodal pairing proves the integral vanishes. Sufficiency also holds for any restricted set of future axes; necessity need not.
\end{proof}
For example, if $s(\lambda+\pi)=-s(\lambda)$ and $|s|\leq1$, then
\begin{equation}
\begin{gathered}
 L(+\mid\lambda)=\frac12+\epsilon\frac{A_\pi(\lambda)}{A(\lambda)+A_\pi(\lambda)}s(\lambda),\\
 0\leq\epsilon\leq\tfrac12,\label{eq:protectedbit}
\end{gathered}
\end{equation}
is nonconstant and protected with $\ell(+)=1/2$. In contrast, at equal widths $g$, summing both outcomes gives the second circular moment (derived in \cref{app:classical})
\begin{equation}
\begin{gathered}
 h_{xy}:=\E_{\rho_{xy}}e^{2i\lambda}=\frac{e^{2ix}+e^{2iy}}{2\cosh(2g)},\\
 \rho_{xy}=\frac{R_AR_B}{Z_{xy}}.\label{eq:moment}
\end{gathered}
\end{equation}
The passive bit $L(+\mid\lambda)=[1+\epsilon\cos2\lambda]/2$ consequently has, for $x=0$,
\begin{equation}
 P(+\mid y)=\frac12+\frac\epsilon2\sech(2g)\cos^2y.\label{eq:leakybit}
\end{equation}
At $g=\epsilon=0.2$ it changes from $0.5925007452$ to $0.5$ as $y$ changes from $0$ to $\pi/2$. A passive exponential timestamp with rate $1+0.2\cos2\lambda$ similarly has survival probabilities at time one of $0.3060094471$ and $0.3732265688$. A scalar depending only on fixed $(x,y)$ cancels in normalization and cannot repair these fixed-setting records.

\begin{corollary}[Continuous exponential-clock rigidity]\label{cor:clock}
For fixed $x,a$, suppose $\tau\mid\lambda$ is exponential with positive continuous rate $v(\lambda)$ on the circle. Its probability distribution is independent of every future axis exactly when $v$ is constant in $\lambda$.
\end{corollary}
\begin{proof}
Apply \cref{eq:antipodal} to $\{\tau>t\}$. With $p=A/(A+A_\pi)\in(0,1)$,
$p(\lambda)e^{-v(\lambda)t}+[1-p(\lambda)]e^{-v(\lambda+\pi)t}=S(t)$.
Impose this first on a dense countable set of times, then use continuity in $t$. Uniqueness of Laplace transforms identifies a common at-most-two-point measure of rates. Positive mixture weights confine every rate, outside a null set, to that fixed finite support. Continuity extends this constraint everywhere and connectedness makes the rate constant. The converse is immediate.
\end{proof}
The rate may depend on already recorded $x,a$; the statement concerns a passive clock coupled to an unresolved source angle.

For any measurable policy $y=f(a,r)$ and kernels satisfying \cref{thm:passive}, specify the corrected history measure
\begin{equation}
 \begin{aligned}
 &\widetilde P(\dd\lambda,a,\dd r,b)\\
 &\quad=\frac{A_a(\lambda)B_{f(a,r),b}(\lambda)}{Z_{x,f(a,r)}}\\
 &\qquad\times L_{x,a}(\dd r\mid\lambda)\dd\lambda.
 \end{aligned}\label{eq:passiveadaptive}
\end{equation}
Antipodal pairing, with the selected $y$ held fixed at each record, gives $\widetilde P(a,\dd r)=\ell_{x,a}(\dd r)/2$, so this law already has mass one. This kernel-measure argument handles arbitrary measurable policies and does not assume a common density for all readouts. It stipulates a history law rather than deriving its physical preparation.

\subsection{A common analyzer, recorded seeds, and readable environments}
The compensated density $\rho_{xy}$ in \cref{eq:moment}, together with normalized responses
\begin{equation}
\begin{gathered}
 T_A(a\mid\lambda,x)=A_a/R_A,\\
 T_B(b\mid\lambda,y)=B_{y,b}/R_B,\label{eq:comp}
\end{gathered}
\end{equation}
reproduces the full hidden law \cref{eq:original}. Summing over $a,b$ shows that $\rho_{xy}$ is forced if both responses and that target law are fixed. Each detector uses the same formula for an incoming spin angle $\phi$:
\begin{equation}
\begin{gathered}
 T_\eta(c\mid\phi,z)=\frac12[1+c v_\eta\cos(\phi-z)],\\
 v_\eta=\sech\eta.\label{eq:analyzer}
\end{gathered}
\end{equation}
Alice receives $\phi=\lambda$, Bob $\phi=\lambda+\pi$. Outcome $c$ prepares $\theta_{z,c}$, which supplies the next use of the same kernel. The response depends on the incoming angle and analyzer parameters, without an additional label distinguishing source inputs from prepared continuations. A normalized reset on that input may be inserted. The source density depends on future axes, so it is not a setting-independent forward prior. A formal canonical potential $-\log(R_AR_B)$ at fixed inputs neither supplies a retrocausal carrier nor defines adaptive boundaries.

Let $U$ be uniform on $[0,1]$, $\delta=\phi-z$, $s(\delta)=\operatorname{sgn}\cos\delta$ with opposite assignments at antipodal zeros, and $t_\eta(\delta)=[1+v_\eta|\cos\delta|]/2$. Define
\begin{equation}
 D_z(\phi,U)=s(\delta)\begin{cases}+1,&U<t_\eta(\delta),\\-1,&U\geq t_\eta(\delta).\end{cases}\label{eq:deterministic}
\end{equation}
Its averaged response is \cref{eq:analyzer}, but it also satisfies $D_z(\phi+\pi,U)=-D_z(\phi,U)$ at every seed. For $I_a(\lambda,U)=\mathbf1_{D_x(\lambda,U)=a}$, one has $I_a(\lambda,U)+I_a(\lambda+\pi,U)=1$. Since $\rho_{xy}$ is $\pi$-periodic,
\begin{equation}
 P(a,\dd U\mid x,y)=\tfrac12\dd U.\label{eq:seed}
\end{equation}
Replacing $y$ by any measurable $f(a,U)$ within the history density preserves this identity and normalization. An independent constant-rate latch clock or an independently archived Bob seed preserves this identity after marginalization over later outcomes. Under feedback, the conditional angle distribution given the earlier record is the restricted, normalized history law; it is not automatically the original fixed-setting prior.

Equality of averaged kernels does not protect every realization. The ordinary inverse-CDF detector, $a=+1$ iff $U<p(\lambda)=[1+v_g\cos\lambda]/2$ at $x=0$, gives
\begin{equation}
\begin{gathered}
 \E[aU\mid\lambda]=p(\lambda)^2-\tfrac12,\\
 \E[aU\mid y]=-\tfrac14+\frac{v_g^2}{8}+\frac{v_g^2\sech(2g)}8\cos^2y.\label{eq:icdf}
\end{gathered}
\end{equation}
The last moment differs by $0.1111214901$ between $y=0$ and $\pi/2$ at $g=0.2$.

\begin{theorem}[Environmental protection]\label{thm:environment}
Append a normalized standard Borel kernel $Q(\dd e\mid\lambda,a,U,x)$ without changing $\rho_{xy}$ or \cref{eq:deterministic}. The joint record $(a,U,E)$ is independent of every future axis if and only if a probability kernel $\ell_{a,U,x}$ satisfies, almost everywhere,
\begin{equation}
 \begin{aligned}
 &I_a(\lambda,U)Q(\dd e\mid\lambda,a,U,x)\\
 &\quad+I_a(\lambda+\pi,U)Q(\dd e\mid\lambda+\pi,a,U,x)\\
 &\qquad=\ell_{a,U,x}(\dd e).
 \end{aligned}\label{eq:environment}
\end{equation}
Then $P(a,\dd U,\dd e)=\dd U\,\ell_{a,U,x}(\dd e)/2$, also under $y=f(a,U,e)$ in the corresponding compensated history rule.
\end{theorem}
\begin{proof}
For fixed $a,U,x$ and an event $F$, subtract the protected marginal from $\int\rho_{xy}I_aQ(F\mid\lambda)\dd\lambda$ and multiply by $Z_{xy}$. As in \cref{thm:passive}, the translates of $R_B$ force the even part of $R_A[I_aQ(F\mid\lambda)-\ell_{a,U,x}(F)/2]$ to vanish. Positivity and $\pi$-periodicity of $R_A$ give \cref{eq:environment}. Disintegration, countable determining classes, a dense countable set of axes, and continuity of source translates yield common null sets. Pairing proves sufficiency and, as a kernel-measure identity before integrating the record, permits substitution of any measurable policy.
\end{proof}
Exactly one indicator in \cref{eq:environment} is one. The readable environmental law must therefore be independent of the incoming angle on the portion producing that fixed $(a,U,x)$; omitting an accessible variable from the recorded tuple does not establish consistency for the joint record.

For a counterexample, export the incoming angle to a blank classical environmental spin and prepare the outgoing spin at $\theta_{x,a}$. Read the exported spin before $y$ is chosen, with \cref{eq:analyzer} and an independent seed, obtaining $d$. At equal widths and $x=0$,
\begin{equation}
\begin{gathered}
 \E[ad\mid y]=\frac{v_g^2}{2}[1+\sech(2g)\cos^2y],\\
 \Delta P(a=d)=\frac{v_g^2\sech(2g)}4.\label{eq:export}
\end{gathered}
\end{equation}
The two endpoint probabilities are $0.962503725953$ and $0.740260745742$ at $g=0.2$. A normalized later reset cannot change this archived statistic.

More generally, suppose a deterministic reversible apparatus has a known or noninvasively archived ready microstate, overwrites the incoming angle by an outgoing state determined by $(x,a)$, and makes all other information needed to invert the operation jointly readable before $y$ is chosen. If reading these outputs leaves the source law unchanged, their protected joint distribution would push forward under the fixed inverse map to a protected distribution of $\lambda$. This contradicts \cref{eq:moment}. At unequal finite positive widths, distinct normalized $R_AR_B$ distributions also exist: equality for every $y$ would make all translates of nonconstant $R_B$ proportional, hence equal by their common integral. This restricted classical obstruction excludes neither inaccessible quantum degrees of freedom nor active readouts that change the source boundary law.

\begin{figure}[tbp]
\centering\includegraphics[width=\columnwidth]{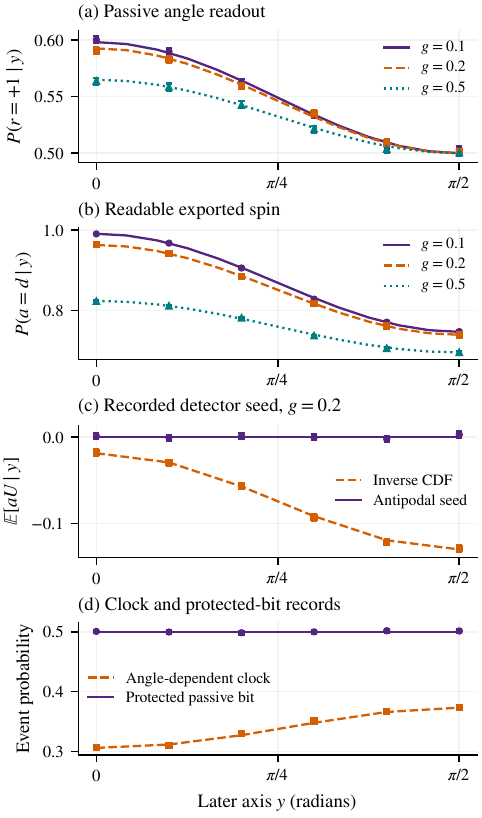}
\caption{Dependence of detector records on the subsequent analyzer axis, with $x=0$. (a) Positive-outcome probability $P(r=+1\mid y)$ of a passive source readout with $\epsilon=0.2$. (b) Agreement probability $P(a=d\mid y)$ between detector and exported-spin outcomes. Panels (a) and (b) use $g\in\{0.1,0.2,0.5\}$. (c) Mean outcome- seed product $\E[aU\mid y]$ for inverse-CDF and antipodal detectors with identical averaged responses. (d) Clock survival probability $P(\tau>1\mid y)$ for rate $1+0.2\cos2\lambda$, and positive-outcome probability of the protected readout in \cref{eq:protectedbit}, with $\epsilon=0.2$ and $s(\lambda)=\sin\lambda$. Panels (c) and (d) use $g=0.2$. Lines connect deterministic predictions; symbols represent $100{,}000$ histories per setting at six equally spaced axes. Error bars are pointwise 95\% Wilson intervals for probabilities and $1.96$ standard errors in (c); some are smaller than the symbols. Readouts on shared histories are correlated. The protected passive record marginalizes the detector seed; protection of their combined record is not asserted.}\label{fig:classical}
\end{figure}

\section{Detector dynamics, controller normalization, and reset}\label{sec:controller}
\subsection{First-event dynamics and the fixed-past limitation}
For a prepared endpoint $(q,b)$, next axis $z$, and width $\eta$, let
\begin{equation}
\begin{gathered}
 k_c=\frac1{\cosh\eta-bc\cos(z-q)},\\
 R=\sum_c k_c=\frac{2\cosh\eta}{\cosh^2\eta-\cos^2(z-q)}.\label{eq:rates}
\end{gathered}
\end{equation}
Two escape rates $\nu k_c$, with $\nu$ in inverse time, generate first-event density $f(c,t)=\nu k_c e^{-\nu Rt}$ and survival $S(t)=e^{-\nu Rt}$. Integration over $t\geq0$ gives $k_c/R$, equal to \cref{eq:analyzer} at $\phi=\theta_{q,b}$. At a finite deadline $t_*$,
\begin{equation}
\begin{gathered}
 P(c,t\leq t_*)=\frac{k_c}{R}(1-e^{-\nu Rt_*}),\\
 P(\varnothing)=e^{-\nu Rt_*}.\label{eq:deadline}
\end{gathered}
\end{equation}
The no-event branch is essential to normalization. Conditioning on detection multiplies earlier branch masses by $1-e^{-\nu R(h)t_*}$ and can disturb the earlier record. The instrument, including the no-event branch, preserves the input record distribution. The use of anomalous-rotation weights as transition rates is a mesoscopic model assumption. Local detailed balance fixes rate ratios, not their full dynamics \cite{maes,seifert}. An explicit barrier model and its finite retention are given in \cref{app:classical}.

Replacing both original arms by normalized detectors and drawing a setting-independent uniform angle gives a different pair law:
\begin{equation}
 P_{\rm fw}(a,b\mid x,y)=\frac14\left[1-\frac{ab}{2}v_{\gamma_A}v_{\gamma_B}\cos(x-y)\right].\label{eq:forward}
\end{equation}
At equal widths $0.2$, its visibility is $0.4805214915$ rather than $0.9250074519$, with optimized coplanar CHSH values $1.3591200205$ and $2.6163161676$. These values characterize the specified forward replacement of the retrocausal source. With $(x,a)$ communicated to a causally later Bob, competing rates proportional to $\widehat W^{ab}_{xy}$ reproduce the desired conditional pair probabilities. An abstract marked process with rate density $\nu A_a(\lambda)B_{y,b}(\lambda)$ has total rate $\nu Z_{xy}/2$ and integrated mark density $2A_aB_{y,b}/Z_{xy}$; drawing $a$ fairly reproduces \cref{eq:original}. Its mark specifies a full source history, not a derived local mechanism.

A normalized later kernel cannot change a fixed earlier hidden distribution $q(\lambda\mid h)$ after all later outputs are summed. Thus it cannot produce two different target distributions \cref{eq:conditional}. At $x=0,a=+1,g=0.2$, their total-variation distance for $y=0,\pi/2$ is $0.5671523677$, so any common $q$ has worst error at least $0.2835761838$. The simpler analytic lower bound is $\sech(0.4)/4$, from the bounded observable $\cos2\lambda$. These bounds constrain approximations based on a future-independent hidden distribution; they do not bound observable signalling in the isolated source model.

\subsection{Gaussian history weights and normalized reset}
Set $r=L/L_{\rm ref}>0$, using one reference mass, and dimensionless coordinates $Q_0,Q_1\in\mathbb R^2$ with fixed normalized ready density $\pi(Q_0)$. The reciprocal raw link
\begin{equation}
 G_{\alpha,r}(Q_1,Q_0)=\frac{r^\alpha}{2\pi}\exp[-r|Q_1-Q_0|^2/2]
\end{equation}
has integrated mass $r^{\alpha-1}$ for every real $\alpha$. At $\alpha=0$ this is $1/r$; at $\alpha=1$ it is a normalized transition. Both have the same reciprocity and ready state. In general $n$ equally scaled quadratic coordinates supply $r^{-n/2}$, explaining why two coordinates give the required inverse power.

For a prepared continuation choose $L=R$ in \cref{eq:rates}; for a later partner's first detection use $L=\widehat Z_{xy}$. The latter needs an interface carrying the relevant earlier setting and source parameters to a timelike later controller. This interface applies to timelike communication; a spacelike implementation or a construction with independent inverse source factors on both source arms would require additional assumptions. With orthogonal outcome feedback the two branch masses have ratio $\varrho^\alpha$, where $\varrho=\widehat Z_{x,x}/\widehat Z_{x,x+\pi/2}=\coth^2\Gamma$. The signed bias and total-variation disturbance are, respectively,
\begin{equation}
\begin{gathered}
 b_\alpha=P(a=+)-\tfrac12=\frac{\varrho^\alpha-1}{2(\varrho^\alpha+1)},\\
 \epsilon_\alpha=|b_\alpha|.\label{eq:gaussianbias}
\end{gathered}
\end{equation}
The absolute value is necessary for negative $\alpha$. At $\Gamma=0.4$ and $\alpha=0,1/2,1,2$, the disturbances are $0$, $0.2246644821$, $0.3738499591$, and $0.4795852858$. Reciprocity does not select the cancelling member.

With $\beta=(k_BT)^{-1}$, specify
\begin{equation}
\begin{gathered}
 H_V=\tfrac12k_BT r|Q_1-Q_0|^2+V(r),\\
 C_V(r)=\frac{e^{-\beta V(r)}}r.\label{eq:energy}
\end{gathered}
\end{equation}
The raw measure is $\dd^2Q_1/(2\pi)$. Cancellation $C_Vr=\text{constant}$ holds precisely for a constant bare potential; constants must also agree between branches. Kinetic energies, degeneracies, and other integrated degrees of freedom must be common factors or included in $V$. The prefactor family corresponds to $V_\alpha=-\alpha k_BT\log r$. Its conditional equilibrium quantities are
\begin{align}
 F_V(r)&=V(r)+k_BT\log r+\text{constant},\\
 \langle-\partial_rH_V\rangle&=-k_BT/r-V'(r),\\
 W_{1\to r}&=V(r)-V(1)+k_BT\log r.\label{eq:work}
\end{align}
At $\alpha=0$ the holding force is $k_BT/r$; at $\alpha=1$ it is zero. Conditional oscillator fluctuations coincide, but backreaction differs. Force data determine $V'$ over a connected range, not separate branch offsets. The reverse quasistatic work cancels the forward work for every $\alpha$, so a zero-work cycle does not select the cancellation.

Ordinary conditional thermal preparation divides by $C_V$ and integrates to one, removing the desired inverse factor. Similarly, a sudden stiffness quench of a normalized two-coordinate Gaussian has $s=|Q|^2/2$ exponentially distributed and
\begin{equation}
 \langle e^{-\beta W}\rangle=\int_0^\infty e^{-s}e^{-(r-1)s}\dd s=1/r.\label{eq:jarzynski}
\end{equation}
This work identity \cite{jarzynski} is a weighted expectation, not an unselected branch probability. A raw Gibbs history ensemble requires its own coupling or boundary-law postulate.

A reset on the complete state $X$, including environmental output $E$, preserves raw branch mass when
\begin{equation}
 \int\mathcal R_{h,u}(\dd X',\dd E\mid X)=1\quad\text{for every }X,h,u.\label{eq:reset}
\end{equation}
Integrating the reset first proves this irrespective of work cost. A display reset can be the restriction of the reversible swap $(m,E=0,H,e)\mapsto(0,E=m,H,e)$, with archive $H$ and preparation interface $e$. Reuse consumes blank capacity or an additional erasure process; it is not restoration of every environment variable. A normalized Gaussian reset
$T_s(Q_2\mid Q_1)=s(2\pi)^{-1}e^{-s|Q_2-Q_1|^2/2}$ appended to $G_{0,r}$ preserves mass $1/r$. Its density at $Q_2=Q_0$ is $s/[2\pi(r+s)]$, whereas a ready window $|Q_2-Q_0|<a_0$ has raw mass
\begin{equation}
 \frac1r\left[1-\exp\left(-\frac{rs a_0^2}{2(r+s)}\right)\right].\label{eq:window}
\end{equation}
An exact endpoint is a density, not a positive-probability event; selecting a window adds a branch-dependent success factor. Finally, two preparations that agree in every variable $\xi$ affecting the next device and its later dynamics have the same response $\int q(\xi)D(o\mid\xi,z)\dd\xi$. Identical distributions over the specified interface variables therefore yield identical device responses.

\section{Quantum instruments and a sharp source-probe tradeoff}\label{sec:quantum}
\subsection{Operational benchmark and a setting-dependent hybrid}
Let $\bm n_x=(\sin x,0,\cos x)$ and let $\ket{s_x}$ be the eigenstate of $\bm n_x\cdot\bm\sigma$ with eigenvalue $s=\pm1$. Define effects and Kraus operators
\begin{align}
 E_a^\gamma(x)&=\tfrac12[\id+a v_\gamma\bm n_x\cdot\bm\sigma],\\
 K_{a,s}^{\gamma,x}&=\sqrt{\frac{1+asv_\gamma}{2}}\ket{a_x}\bra{s_x},\\
 \mathcal I_a^{\gamma,x}(\tau)&=\sum_s K_{a,s}^{\gamma,x}\tau K_{a,s}^{\gamma,x\dagger}
 =\Tr[E_a^\gamma(x)\tau]\ket{a_x}\bra{a_x}.\label{eq:instrument}
\end{align}
Summing $K_{a,s}^\dagger K_{a,s}$ over $a,s$ gives $\id$. This is an entanglement-breaking measure-and-prepare instrument \cite{horodecki}; a pure input at $\phi$ gives \cref{eq:analyzer} and the correct outgoing eigenstate. For fixed widths choose the Werner state \cite{werner}
\begin{equation}
\begin{gathered}
 \omega_w=w\ket{\psi^-}\bra{\psi^-}+(1-w)\id_4/4,\\
 w=\frac{\sech(\gamma_A+\gamma_B)}{v_{\gamma_A}v_{\gamma_B}}
 =\frac1{1+\tanh\gamma_A\tanh\gamma_B}.\label{eq:werner}
\end{gathered}
\end{equation}
Here $\ket{\psi^-}=(\ket{01}-\ket{10})/\sqrt2$. Thus $1/2<w<1$, and its correlation tensor $-w\id_3$ gives
$\Tr[(E_a^{\gamma_A}(x)\otimes E_b^{\gamma_B}(y))\omega_w]=P_0(a,b\mid x,y)$.
At equal widths $0.2$, $w=0.962503725953$. A single unchanged source with freely varied first-segment widths is a different problem; in the stated benchmark those widths fix $w$.

The local isometry
\begin{equation}
\begin{gathered}
 V_x\ket\psi=\sum_{a,s}K_{a,s}^{\gamma,x}\ket\psi\otimes\ket{a,s}_E,\\
 V_x^\dagger V_x=\id,\label{eq:dilation}
\end{gathered}
\end{equation}
extends to a unitary with a ready ancilla \cite{stinespring}. An explicit circuit and source purification appear in \cref{app:quantum}. Orthogonal pointer labels can be copied to an archive; a swap with a blank register resets the display while preserving displaced information. Reversing the dilation requires the correlated output, not a fresh blank environment.

For an unnormalized conditional state $\sigma_h$ after any earlier instrument, a later trace-preserving channel selected by $u=f(h)$ obeys
\begin{equation}
 \Tr\Lambda_{f(h)}(\sigma_h)=\Tr\sigma_h.\label{eq:tp}
\end{equation}
This standard quantum causality property includes earlier local environment measurements, retained reset outputs, and adaptive continuation \cite{combs}. A selected reset outcome is trace decreasing and need not obey it. For the specific environment label in \cref{eq:instrument},
$P(a,s\mid x,y)=[1+asv_{\gamma_A}]/4$, independent of $y$.

A closer connection to the source angle uses the separable but setting-dependent family
\begin{equation}
\begin{gathered}
 \Omega_{xy}=\int_0^{2\pi}\rho_{xy}(\lambda)\ket\lambda\bra\lambda\otimes\ket{\lambda+\pi}\bra{\lambda+\pi}\dd\lambda,
 \\
\ket\lambda=\cos(\lambda/2)\ket0+\sin(\lambda/2)\ket1.\label{eq:hybrid}
\end{gathered}
\end{equation}
Conditioned on $\lambda$, the instruments yield $T_AT_B$, preserving the full isolated law. Antipodal pairing gives $\Tr_A\Omega_{xy}=\Tr_B\Omega_{xy}=\id_2/2$. Consequently any fixed earlier single-arm apparatus with effect $F_h$, including its quantum ancillas and environment, has $P(h\mid x,y)=\Tr(F_h\id_2/2)$. The branchwise rule $y=f(h)$ preserves this local record measure as well. Its joint source state nevertheless depends on settings.

Suppose this same source assignment is used with a joint probe before the later settings. For equal widths, write $q=\sech(2g)$, fix $x=0$, and choose $F_{\rm even}=(\id_4+Z\otimes Z)/2$, where $X,Y,Z$ denote Pauli matrices. Since the conditional $zz$ correlation is $-\cos^2\lambda$,
\begin{equation}
 P({\rm even}\mid y)=\tfrac14-\tfrac q4\cos^2y.\label{eq:parity}
\end{equation}
At $g=0.2$ the endpoint probabilities are $0.018748137024$ and $0.25$. More precisely,
\begin{equation}
\begin{gathered}
 \Omega_{00}-\Omega_{0,\pi/2}=\frac q8(X\otimes X-Z\otimes Z),\\
 D(\Omega_{00},\Omega_{0,\pi/2})=q/4,\label{eq:hybriddistance}
\end{gathered}
\end{equation}
where $D(\alpha,\beta)=\tfrac12\|\alpha-\beta\|_1$; the parenthesized operator has eigenvalues $2,-2,0,0$.

If all earlier joint effects are allowed and a source state is assigned independently of the inserted probe, invariance of every earlier probability forces that state to be setting independent: equality against every rank-one projector determines a Hermitian matrix. If it is also a positive product-state mixture with local detector responses, its decomposition supplies a setting-independent local hidden variable and obeys CHSH \cite{chsh}. Thus the unchanged separable hybrid cannot universally retain both all such records and the target correlations when $\sech\Gamma>1/\sqrt2$. This argument does not apply to a source assignment changed by an inserted probe, which we consider next.

\subsection{Calibrated weak probes: exact bound, attainability, and scope}
Let $\Omega_0=\Omega_{00}$, $\Omega_1=\Omega_{0,\pi/2}$, and $A=Z\otimes Z$. Their expectations are $m_0=-(1+q)/2$ and $m_1=-1/2$. A normalized weak parity instrument has outcomes $r=\pm1$ and
\begin{equation}
\begin{gathered}
 F_r^{(k)}=\tfrac12(\id_4+rkA),\\
 M_r^{(k)}=\sqrt{F_r^{(k)}},\\
0\leq k\leq1.\label{eq:weak}
\end{gathered}
\end{equation}
Both outcomes are retained. Its earlier plus probabilities on the unchanged hybrid are
\begin{equation}
 \begin{gathered}
 P_0(+)=\tfrac12-\tfrac{k(1+q)}4,\\
 P_1(+)=\tfrac12-\tfrac k4,\quad \delta_{\rm raw}(k)=kq/4.
 \end{gathered}\label{eq:rawsignal}
\end{equation}
Both source states commute with $A$, so the nonselective instrument leaves them unchanged. The difference between the archived record distributions therefore persists despite the absence of nonselective state disturbance. Two unitary implementations are derived in \cref{app:quantum}; unsharp measurements of this kind are standard \cite{choudhary}.

\begin{theorem}[Calibrated-probe bound]\label{thm:weak}
Let $\Omega_0,\Omega_1$ be density operators, $A$ Hermitian with $\|A\|_\infty\leq1$, and $0<k\leq1$. Permit source responses $\sigma_j(k)$, but suppose the actual earlier probabilities in both contexts are $\Tr[F_r^{(k)}\sigma_j(k)]$ with $F_r^{(k)}=(\id+rkA)/2$. Set
\begin{equation}
 \begin{gathered}
 \epsilon=\max_{j=0,1}D(\sigma_j(k),\Omega_j),\\
 \Delta m=|\Tr[A(\Omega_0-\Omega_1)]|,\\
 \delta=|P_0(+)-P_1(+)|.
 \end{gathered}
\end{equation}
Then
\begin{equation}
\begin{gathered}
 \delta\geq\frac k2\max\{0,\Delta m-4\epsilon\},\\
 \epsilon\geq\max\{0,\Delta m/4-\delta/(2k)\}.\label{eq:bound}
\end{gathered}
\end{equation}
\end{theorem}
\begin{proof}
Trace-norm duality gives $|\Tr[A(\sigma_j-\Omega_j)]|\leq2D(\sigma_j,\Omega_j)\leq2\epsilon$. The triangle inequality gives
$|\Tr[A(\sigma_0-\sigma_1)]|\geq\Delta m-4\epsilon$. Multiply by $k/2$, impose nonnegativity, and rearrange.
\end{proof}
If $\Delta m>0$, exact protection at every $k>0$ is incompatible with both $\sigma_j(k)\to\Omega_j$ in trace norm as $k\downarrow0$. For the hybrid pair, $\Delta m=q/2$, so $\epsilon\geq q/8=0.115625931488$ at $g=0.2$, irrespective of nonzero strength and without any Bell-violation requirement.

\begin{proposition}[Sharp tradeoff for the hybrid pair]\label{prop:sharp}
At fixed $k>0$, the minimum worst-context trace-distance change compatible with signal tolerance $\delta\leq\delta_{\max}$ is
\begin{equation}
 \epsilon_{\min}(\delta_{\max},k)=\max\{0,q/8-\delta_{\max}/(2k)\}.\label{eq:sharp}
\end{equation}
The allowed source states can remain separable mixtures of opposite-spin product states.
\end{proposition}
\begin{proof}
Necessity is \cref{eq:bound}. For $0\leq t\leq1/2$, take
$\sigma_0=(1-t)\Omega_0+t\Omega_1$ and $\sigma_1=t\Omega_0+(1-t)\Omega_1$.
By \cref{eq:hybriddistance}, $\epsilon=tq/4$, while linearity gives $\delta=(1-2t)kq/4$. This saturates the full decreasing boundary. Choose $t=0$ if $\delta_{\max}\geq kq/4$; otherwise take $t=(1-4\delta_{\max}/(kq))/2$. Convexity preserves separability and the stated product-state support.
\end{proof}
\begin{figure}[tbp]
\centering\includegraphics[width=\columnwidth]{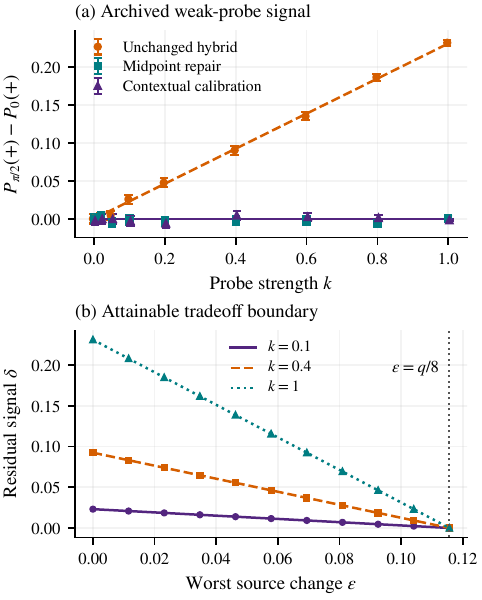}
\caption{Weak-probe record dependence and source modification at $g=0.2$, $x=0$, and $q=\sech(2g)$. (a) Signed probability difference $P(+\mid y=\pi/2)-P(+\mid y=0)$ versus strength $k$. Lines predict $kq/4$ for the unmodified hybrid source and zero for the midpoint source and contextual calibration, which use distinct source and calibration assumptions. Symbols show $50{,}000$ records per independently sampled context, with pointwise 95\% normal intervals of $1.96$ two-sample standard errors. Horizontal offsets separate coincident strengths. (b) Attainable residual signal $\delta$ versus maximum source trace-distance change $\epsilon=\max_jD(\sigma_j,\Omega_j)$. Lines give $\max\{0,kq/4-2k\epsilon\}$ for $k=0.1,0.4,1$; symbols are deterministic evaluations of interpolated source states. The vertical line marks the minimum change for exact protection, $\epsilon=q/8\simeq0.115626$, for $k>0$. The bound concerns the specified two-context calibrated probe.}\label{fig:weak}
\end{figure}

At exact protection $t=1/2$, both states equal the midpoint, whose distance from each original is $q/8$. This optimal source modification is discontinuous at zero strength if the two no-probe states are to be recovered. It preserves the specified probe record for the two contexts, without requiring preservation of every isolated detector statistic or extending to arbitrary axes and instruments. For $k=0.1$ and $\delta_{\max}=0.001$, the minimum is $0.110625931488$. A fixed absolute signal tolerance becomes uninformative at sufficiently small $k$, unlike exact protection.

The calibration assumption states that the density matrix with these fixed effects predicts the actual probabilities in each future context. Calibration on ordinary forward preparations alone does not establish it in a retrocausal extension. With effect uncertainty $\|G_{+,j}-F_+^{(k)}\|_\infty\leq\eta$, the necessary bound becomes
$\epsilon\geq\max\{0,\Delta m/4-(\delta+2\eta)/(2k)\}$.
Errors $o(k)$ preserve the exact-protection continuity obstruction; errors of order $k$ can cancel the leading signal. Small $k$ also makes numerical rearrangement of a measured probability difference ill-conditioned; the stable forward inequality should be checked directly.

Smooth positive histories alone do not imply the calibrated obstruction. For any density $\rho_u$, real $|f|\leq1$, and $m_u=\E_{\rho_u}f$, define for $0\leq k<1$
\begin{equation}
 \widetilde P_u(\dd\lambda,r)=\rho_u(\lambda)\frac{1+rkf(\lambda)}{2(1+rkm_u)}\dd\lambda.\label{eq:smooth}
\end{equation}
Each branch has mass $1/2$, including a branchwise policy $u=u(r)$. At fixed $u$, summing and dividing gives
\begin{equation}
\begin{gathered}
 \widetilde\rho_u=\rho_u\frac{1-k^2m_uf}{1-k^2m_u^2},\\
 \widetilde L_{k,u}(r\mid\lambda)=\frac12\left[1+rk\frac{f(\lambda)-m_u}{1-k^2m_uf(\lambda)}\right].\label{eq:smoothresponse}
\end{gathered}
\end{equation}
The source changes by $O(k^2)$ in total variation, but the normalized response changes by $O(k)$ from $(1+rkf)/2$ and depends on the future context at fixed $\lambda$. This evades the response assumption; the inverse branch factor was chosen to protect the record and is not a derived interaction. In contrast, substituting the unchanged hybrid into the policy $y(+)=0,y(-)=\pi/2$ produces mass $1-kq/4$; reversing the policy gives $1+kq/4$. At $g=0.2,k=0.1$ the first is $0.976874813702$. Global renormalization or a later normalized reset does not establish policy-independent records.

\section{A contextual history embedding and its access obstruction}\label{sec:contextual}
\subsection{Paired apparatus, backward fields, and finite adaptive consistency}
Consider finite experiments with one preparation region and one or two outgoing system arms. Joint source probes occur before separation; each arm permits finite local sequences, retained local environments, and classical feed-forward along causal paths. Independent multiple-source networks and recombination are outside the domain. At each setting $u$ and outcome $o$, specify a nonnegative raw hidden transition $W_{u,o}(\xi,\dd\xi')$ and a completely positive operation $\mathcal J_{u,o}$; the sum over $o$ is trace preserving. Their pairing is model structure, not a consequence of the original source law. The source carries a normalized reference measure $\mu$ and a fixed positive unit-trace calibration tensor $\omega$. The reference $\mu$ is a factor in the history measure, not an independently fixed physical forward source marginal.

\textit{Operational-access postulate.} Admissible records are outcomes of the declared quantum instruments, including instruments on their local environments, clocks, registers, archives, and reset outputs, with raw partners satisfying the support condition below. Copies of the hidden angle or of the added response fields are not implicitly allowed. A physical theory must justify this restriction or specify consistent disturbance when a forbidden readout is attempted.

For the actual setting and outcome, propagate a likelihood $\ell\geq0$ and positive quantum effect $E$ backwards:
\begin{equation}
\begin{gathered}
 \ell_{\rm in}(\xi)=\int W_{u,o}(\xi,\dd\xi')\ell_{\rm out}(\xi'),\\
 E_{\rm in}=\mathcal J_{u,o}^*(E_{\rm out}),\label{eq:recurrence}
\end{gathered}
\end{equation}
where $\Tr[\mathcal J(\tau)E]=\Tr[\tau\mathcal J^*(E)]$. Terminal unused outputs have $\ell=1,E=\id$. At the source split, arm likelihoods multiply and effects tensor: $\ell_S=\ell_A\ell_B$, $E_S=E_A\otimes E_B$. Earlier source probes propagate these through their own $W$ and $\mathcal J$. Environments remain in the local state until their allowed records are resolved. A complete history $h$ fixes causal branch settings, so these finite backward recurrences uniquely determine the fields.

The field measure is the deterministic pushforward of the raw trajectory measure under this solution. In a finite hidden-state discretization, this amounts to normalized Dirac constraints, not an extra field partition factor. The circle model requires function-valued fields. Adjacent-path recurrences and a common source factor supply a factor-graph description, but do not establish finite material carriers, time-reversal invariance, or Lorentz-covariant dynamics. Backward quantum effects also occur as inferential objects \cite{gammelmark}; an inference formula does not itself make an effect a physical field.

Define branch masses, source factor, and formal potential by
\begin{equation}
\begin{gathered}
 d_h=\int\ell_{S,h}(\xi)\mu(\dd\xi),\\
 n_h=\Tr(\omega E_{S,h}),\\
 C_h=\frac{n_h}{d_h},\\
 V_{S,h}=\log d_h-\log n_h.\label{eq:factor}
\end{gathered}
\end{equation}
The numerator and denominator use the same record reference measure. Their ratio is dimensionless and is a raw history factor of the type distinguished from normalized thermal preparation in \cref{sec:controller}.

\begin{theorem}[Consistency of the declared record algebra]\label{thm:embedding}
Fix the initial $\mu,\omega$, the earlier apparatus, and the pairing rules. Let $R_h(\dd\zeta)$ be the raw hidden-trajectory measure with mass $d_h$, including the deterministic field solution. Assume all raw branch masses are finite and $d_h>0$ whenever $n_h>0$. Then
\begin{equation}
 P(h,\dd\zeta)=n_h\frac{R_h(\dd\zeta)}{d_h}\label{eq:embedding}
\end{equation}
defines a normalized nonnegative law for every finite causal adaptive experiment in the stated domain. Its complete operational record distribution is $n_h$, and each earlier record marginal is unchanged by subsequent causal choices. Omit branches with $d_h=n_h=0$ and assign zero measure when $d_h>0=n_h$. A branch with $d_h=0<n_h$ is inadmissible.
\end{theorem}
\begin{proof}
Integration gives $P(h)=n_h$. For an explicit finite adaptive-tree proof, let $\tau_v$ be the unnormalized operational state at a node $v$ determined by its earlier prefix. If the chosen complete instrument is $\{\mathcal J_{u(v),o}\}_o$, child states are $\tau_{vo}=\mathcal J_{u(v),o}(\tau_v)$ and $\sum_o\Tr\tau_{vo}=\Tr\tau_v$. Induct from leaves toward the root: the total mass of all leaves below any node is its incoming trace, independent of choices further down that subtree. Branches of different depths can be treated uniformly by extending a stopped branch with identity operations. The root trace is $\Tr\omega=1$; stopping elimination at an earlier cut proves preservation of the joint prefix record. The adjoint recurrences give the same leaf probability $\Tr\tau_h=\Tr(\omega E_{S,h})$. For continuous records, require measurable instrument densities or kernels, measurable policy maps, and the support and integrability assumptions almost everywhere relative to the record reference measure. Nonnegativity permits the corresponding iterated integrals by Tonelli's theorem.
\end{proof}
The preparation must remain fixed when comparing future policies. The theorem protects operational prefixes, not the marginal law of hidden past trajectories. Equation~\eqref{eq:embedding} embeds the specified quantum law as the product of its record probability and a normalized conditional hidden measure. Altering $W$ while retaining its paired quantum operation and support may change hidden histories without changing any admitted record. The quantum probabilities, source tensor, and access restrictions remain assumptions of the construction.

\subsection{Recovery of the isolated source, probe response, and field accessibility}
Use the same raw analyzer at every incoming angle,
\begin{equation}
 W_{\gamma,z,a}(\phi,\dd\phi')=p_\gamma(\theta_{z,a}-\phi)\delta_{\theta_{z,a}}(\dd\phi'),\label{eq:rawanalyzer}
\end{equation}
paired with \cref{eq:instrument}. For the pair choose opposite source angles, $\mu(\dd\lambda)=\dd\lambda/(2\pi)$, and $\omega=\omega_w$ with fixed first-segment widths. Then
\begin{equation}
\begin{gathered}
 d_{ab}=\frac1{2\pi}\int A_aB_{y,b}\dd\lambda,\\
 n_{ab}=\frac14[1-ab\sech\Gamma\cos(x-y)]
       =\frac{2\pi}{Z_{xy}}d_{ab}.\label{eq:recovery}
\end{gathered}
\end{equation}
The ratio is independent of outcomes, so \cref{eq:embedding} exactly recovers \cref{eq:original}, including its angle dependence. Additional source-path variables are preserved when their conditional raw path law is retained. The numerical tests evaluate the endpoint angle-outcome law. For a prepared input set $\mu=\delta_\phi$, $\omega=\ket\phi\bra\phi$. The same analyzer has $d_a=p_\gamma(\theta_{z,a}-\phi)$, $n_a=d_a/\sum_c d_c$, and $C_a=1/\sum_c d_c$, giving \cref{eq:analyzer} with the correct prepared output. Continuation widths may vary, while matching the original pair with one fixed $\omega_w$ retains fixed first-segment widths.

For positive qubit arm effects $E_j=(t_j\id+\bm q_j\cdot\bm\sigma)/2$, positivity gives $|\bm q_j|\leq t_j$, and
\begin{equation}
 n=\frac{t_At_B-w\bm q_A\cdot\bm q_B}{4}\geq\frac{1-w}{4}t_At_B\geq0.\label{eq:bilinear}
\end{equation}
This is an explicit source functional with one fixed coefficient $w$; quantum calibration remains an input.

For equal first-segment widths $\gamma_A=\gamma_B=g$, insert \cref{eq:weak} in the preparation region, with $0\leq k<1$, paired to
$W_{k,r}(\lambda,\dd\lambda')=L_r(\lambda)\delta_\lambda(\dd\lambda')$, where $L_r(\lambda)=[1+rk\cos2\lambda]/2$. For the complete record $(r,a,b)$,
\begin{align}
 d_{rab}&=\frac1{2\pi}\int L_rA_aB_{y,b}\dd\lambda,\\
 n_{rab}&=\Tr\!\left[\omega_wM_r^\dagger(E_a^g(x)\otimes E_b^g(y))M_r\right],\\
 P(\dd\lambda,r,a,b)&=\frac{n_{rab}}{d_{rab}}\frac{L_r(\lambda)A_a(\lambda)B_{y,b}(\lambda)}{2\pi}\dd\lambda.\label{eq:contextprobe}
\end{align}
Positive raw factors ensure support. Summing future detector outcomes gives
\begin{equation}
 P(r)=\tfrac12(1-rkw),\label{eq:contextrecord}
\end{equation}
independent of both later axes, even if chosen separately on each $r$ branch. The full $n_{rab}$ also specifies the conditioned continuation. At $k=0$, $M_r=\id_4/\sqrt2$, $L_r=1/2$, and summing $r$ recovers \cref{eq:original}. For fixed positive widths and axes the denominators stay positive near zero, so the hidden law converges continuously. Quantum simulations of the operational instrument also include $k=1$; the simple strict-positivity argument for its raw partner is stated only for $k<1$.

The operational source in this construction is $\omega_w$, which determines the source-probe statistics independently of the future-dependent hidden product-state mixture $\Omega_{xy}$. The construction therefore uses a different calibration assumption from that of \cref{thm:weak}. Admitted environment refinements, normalized clocks, archives, and reset have explicit partners in \cref{app:partners}.

\begin{figure}[tbp]
\centering\includegraphics[width=\columnwidth]{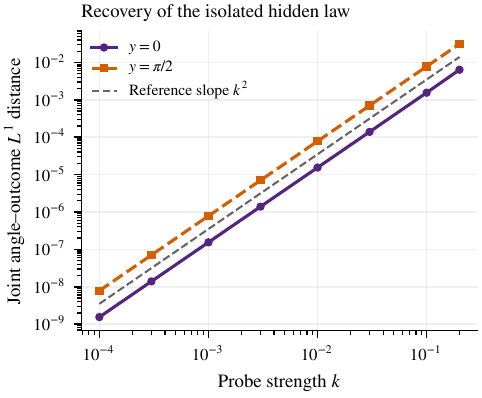}
\caption{Recovery of the isolated hidden source law as probe strength vanishes. The joint angle--outcome distance is $\sum_{a,b}\int_0^{2\pi}|p_k(\lambda,a,b)-p_0(\lambda,a,b)|\,\dd\lambda$, where $p_k=\sum_rP_k(\lambda,r,a,b)$ marginalizes the probe outcome in \cref{eq:contextprobe}, and $p_0$ is the isolated density in \cref{eq:original}. This $L^1$ distance equals twice total variation. Points use deterministic integration on $32{,}768$ angle-grid points, with $g=0.2$, $x=0$, and $y=0,\pi/2$; solid lines connect them. At $k=10^{-4}$, the respective distances are $1.54\times10^{-9}$ and $7.82\times10^{-9}$. The dashed line indicates $k^2$ scaling; the $y=\pi/2$ dependence is exactly quadratic in this construction. The metric concerns the hidden law, not observable probe-record differences. Exact recovery at $k=0$ is checked separately from the logarithmic plot.}\label{fig:convergence}
\end{figure}

For an otherwise isolated future analyzer, the source-side likelihood is $\ell_{x,a}(\phi)=p_g(\theta_{x,a}-\phi)$. Its normalized second harmonic is
\begin{equation}
 \chi_2(\ell)=\frac{\int\ell(\phi)e^{2i\phi}\dd\phi}{\int\ell(\phi)\dd\phi}
 =e^{-2g}e^{2ix}.\label{eq:carrier}
\end{equation}
Neither outcome $a$ nor positive rescaling changes it. An earlier device reading $\operatorname{sgn}\Re\chi_2$ while leaving the downstream recurrence intact therefore has plus probability one at future $x=0$ and zero at $x=\pi/2$. No normalized reweighting on those supports can protect its record. At $g=0.2$ the carrier values are $\pm0.670320046036$. A quantum calibration giving positive weight to a raw-forbidden readout would instead violate $d_h>0$ whenever $n_h>0$.

This support obstruction demonstrates that the access restriction is necessary for the stated embedding. A physical realization must account for either the restriction on field measurements or a consistent modification of the field equations induced by measurement. When the fields serve only as a representation of probabilities, the construction remains a statistical embedding, with a physical retrocausal carrier unspecified.

\section{Numerical experiments and discussion}\label{sec:experiments}
\subsection{Simulation methods and numerical results}
The numerical analysis combines deterministic integration, matrix calculations, and Monte Carlo simulations implemented in the accompanying Python programs. All data are generated from the specified models. Supplementary numerical documentation provides the simulation configurations, estimates, uncertainties, software versions, data formats, and limitations. The accompanying computational materials include compressed trial records, summary CSV files, JSON diagnostics, source code, and additional numerical cross-checks.

The continuous-angle experiment uses equal widths $g\in\{0.1,0.2,0.5\}$, $x=0$, and six equally spaced axes $y\in[0,\pi/2]$, with $100{,}000$ accepted histories per cell. Sampling from $\rho_{xy}$ uses a continuous-angle rejection scheme with proposal density $R_A/2$, an equal mixture of two wrapped-Cauchy distributions, and acceptance probability $R_B/\max R_B$, where $\max R_B=\coth(g)/\pi$. The accepted density is exactly proportional to $R_AR_B$. Independent uniform variates generate the detector outcomes and specified readouts; the angle, seeds, outcomes, and timestamp are retained. The dataset contains $1{,}800{,}000$ accepted histories from $22{,}028{,}513$ proposals. The recorded readouts correspond to alternative apparatus extensions. The consistency results apply separately to $(a,r)$ for the passive kernel and $(a,U)$ for the antipodal detector; consistency of their combined tuple would require an additional joint analysis.

The competing-clock simulations comprise $360{,}000$ trajectories across 18 parameter combinations. Both exponential waiting times are sampled, and the first event is retained, with no-event branches included for finite deadlines. The Gaussian reset simulations comprise $420{,}000$ two-coordinate trajectories across 21 parameter combinations and evaluate normalized conditional displacements and reset-window probabilities. Raw partition factors and return densities are verified by independent quadrature; normalized Gaussian samples alone cannot establish inverse raw history mass. Periodic-grid refinement and adaptive quadrature check source identities, unequal widths, first-event formulas, controller forces, work averages, and reset integrals.

The quantum experiment samples $r,a,b$ sequentially from the post-measurement density operators, using the unchanged hybrid, its common midpoint source, and the contextual Werner calibration. At $g=0.2,x=0$, two future contexts $y=0,\pi/2$ and nine strengths $k\in\{0,0.02,0.05,0.1,0.2,0.4,0.6,0.8,1\}$ give 54 cells of $50{,}000$ trials, or $2{,}700{,}000$ complete records. Six further adaptive cells of $50{,}000$ trials give $300{,}000$ records. For the unchanged hybrid, adaptive sampling uses the globally normalized branchwise history, with the raw mass recorded separately to quantify its deviation from unity. This procedure samples a context-dependent history law rather than a setting-independent forward circuit. Angle-quadrature source matrices, Kraus maps, explicit unitary meters and detector circuits, and direct versus backward-effect evaluation provide deterministic cross-checks.

The exhaustive policy enumeration uses the contextual calibration at $g=0.2$, source-probe strength $k=0.3$, Alice axis $0.237$ radians, and a measured Kraus-environment rotation of $0.413$ radians. The earlier record $h=(r,a,e)$ has eight values. Each selects one of four Bob axes $\{0,\pi/4,\pi/2,3\pi/4\}$. Enumerating all $4^8=65{,}536$ policies and summing all four Bob Kraus terms evaluates $2{,}097{,}152$ complete branch weights. The largest earlier-record residual is $1.11\times10^{-16}$ and the largest normalization residual $6.66\times10^{-16}$. The enumeration covers this finite policy class; the extension to arbitrary measurable policies and the continuum source follows from the analytical results.

\begin{table}[tbp]
\centering
\footnotesize
\setlength{\tabcolsep}{3pt}
\renewcommand{\arraystretch}{1.08}

\caption{Residuals of deterministic numerical checks. Except for the
relative Gaussian partition error, entries are absolute deviations in the
indicated probability, moment, integral, density, or matrix-element identity.
These floating-point residuals quantify agreement between numerical
evaluations and their reference identities; rigorous discretization error
bounds require a separate analysis. The largest residual among the reported
quantum checks is $6.44\times10^{-15}$. The carrier harmonic residual measures
agreement with the analytical expression; the associated operational signal
has unit magnitude.}
\label{tab:residuals}

\begin{tabular}{@{}ll@{}}
\toprule
Deterministic check & Largest residual \\
\midrule

Source pair probabilities
&
$6.66\times10^{-16}$
\\
\addlinespace

\shortstack[l]{Source second circular\\moment}
&
$6.77\times10^{-16}$
\\
\addlinespace

Protected passive readout
&
$7.77\times10^{-16}$
\\
\addlinespace

Source normalization
&
$1.11\times10^{-15}$
\\
\addlinespace

\shortstack[l]{81 classical adaptive\\seed-bin policies}
&
$1.67\times10^{-16}$
\\
\addlinespace

\shortstack[l]{Gaussian partition\\(relative)}
&
$4.44\times10^{-16}$
\\
\addlinespace

\shortstack[l]{Gaussian reset window\\(absolute)}
&
$4.44\times10^{-16}$
\\
\addlinespace

Exponential work integral
&
$1.78\times10^{-15}$
\\
\addlinespace

\shortstack[l]{Quantum Kraus\\completeness}
&
$5.55\times10^{-16}$
\\
\addlinespace

Four-qubit detector circuit
&
$2.22\times10^{-16}$
\\
\addlinespace

\shortstack[l]{Quantum source-angle\\matrix quadrature}
&
$6.38\times10^{-15}$
\\
\addlinespace

\shortstack[l]{Quantum source local\\marginals}
&
$6.44\times10^{-15}$
\\
\addlinespace

Sharp tradeoff equality
&
$7.08\times10^{-16}$
\\
\addlinespace

\shortstack[l]{Contextual zero-coupling\\full density}
&
$1.11\times10^{-15}$
\\
\addlinespace

\shortstack[l]{65,536 adaptive policies:\\earlier marginal}
&
$1.11\times10^{-16}$
\\
\addlinespace

\shortstack[l]{Carrier harmonic\\(field readout)}
&
$4.77\times10^{-16}$
\\

\bottomrule
\end{tabular}
\end{table}

Monte Carlo uncertainty is reported separately from \cref{tab:residuals}. Bernoulli probabilities use pointwise Wilson 95\% intervals; other means and differences use stated standard-error intervals. The classical suite's largest standardized discrepancy is $3.544$ over 942 correlated comparisons. In the 54 quantum plus-probability cells, 48 Wilson intervals contain the target, and the largest absolute null-standardized discrepancy is $2.459$; none exceeds the stated Bonferroni normal threshold $3.312$ for those 54 cells. These comparisons quantify agreement with the model predictions within sampling uncertainty. Exact independence and the domain of calibration follow from the analytical assumptions and results, while the reported pointwise intervals do not provide simultaneous coverage of all quantities. The smallest weak signals are not reliably resolved at the selected sample size, as indicated by the uncertainty intervals in \cref{fig:weak}.

The simulations use NumPy's PCG64 generator with random seeds 2026091701 (classical) and 2026091702 (quantum). The computational environment comprises Python 3.12.14, NumPy 2.3.5, and SciPy 1.17.0. The programs accept model parameters and trial counts, save trial-level records, and generate numerical summaries. A separate plotting program generates the figures from the saved data. Execution commands, software versions, definitions of recorded variables, numerical cross-checks, and file checksums are supplied with the computational materials. The manuscript source and vector figures are provided separately from the simulation data archive.

\subsection{Interpretation and remaining physical requirements}
The source-specific antipodal identities distinguish hidden-angle dependence that cancels in every admitted earlier record from dependence exposed by a clock, seed, or environment. The quantum benchmark shows that the desired pair and continuation statistics are operationally consistent at fixed widths. The sharp probe theorem quantifies the minimum source change required under fixed probe calibration. This minimum remains nonzero at every nonzero coupling, even as the uncorrected signal decreases. The result applies to apparatus extensions satisfying the stated calibration assumptions.

The contextual embedding recovers the isolated hidden law and preserves the declared operational records using a specified quantum probability law. Its hidden conditional measures may depend on future apparatus choices, while its calibration tensor, raw partners, and access rules are model inputs. Direct carrier readout violates record consistency when the fields are treated as noninvasively accessible classical coordinates: the analytical support calculation gives earlier readout probabilities of zero and one for the two future settings.

A microscopic physical realization would require a source-field interaction, an ensemble with specified boundary conditions, forward and reversed preparation procedures, and a mechanism determining the accessible records or their measurement-induced disturbance. The Gaussian history weights, normalized resets, backward effects, and reversible dilations analyzed here constrain aspects of such a realization but do not determine its full dynamics. The present results establish mathematical characterizations and a conditional statistical embedding within the stated apparatus classes. Microscopic implementation and extension to unrestricted multiple-source networks remain open problems.

\paragraph{Author contributions and use of AI tools.}
The author used OpenAI's ChatGPT to assist with implementation of numerical verification code, and manuscript preparation. The author takes responsibility for the final arguments, computations, interpretations, and conclusions.
\clearpage
\onecolumn
\appendix
\section{Supporting classical derivations and implementation constraints}\label{app:classical}
For equal widths, the antipodal sum of \cref{eq:kernel} is
\begin{equation}
 R_A(\lambda)=\frac{\sinh(2g)}{\pi[\cosh(2g)-\cos2(\lambda-x)]}.
\end{equation}
To derive \cref{eq:moment}, put $C=\cosh(2g)$, $d=x-y$, $t=2\lambda-x-y$, and $U_x=C-\cos(t-d)$, $U_y=C-\cos(t+d)$. The unnormalized density $(U_xU_y)^{-1}$ is even in $t$, so $\E\sin t=0$. Convolution gives
\begin{equation}
 J=\int_0^{2\pi}\frac{\dd\lambda}{U_xU_y}=\frac{2\pi C}{\sinh(2g)(C^2-\cos^2d)}.
\end{equation}
Integrating $U_x+U_y=2C-2\cos d\cos t$ against that normalized density yields
\begin{equation}
 2C-2\cos d\E\cos t=\frac{4\pi}{\sinh(2g)J}=\frac{2(C^2-\cos^2d)}C.
\end{equation}
For $\cos d\ne0$, this gives $\E\cos t=\cos d/C$; continuity supplies the remaining case. Multiplying by $e^{i(x+y)}$ gives \cref{eq:moment}.

For an explicit mesoscopic pointer, take energies $E_0=0$, $E_+=E_-=-\Delta$ with $\Delta>0$. For attempt rate $\nu_{\rm att}>\nu\max_c k_c$, choose barriers
\begin{equation}
 B_c=k_BT\log\frac{\nu_{\rm att}}{\nu k_c}>0.
\end{equation}
Activated rates are $r_{0c}=\nu k_c$ and $r_{c0}=\nu k_ce^{-\beta\Delta}$, so $\log(r_{0c}/r_{c0})=\beta(E_0-E_c)$ as required by local detailed balance. With fixed inputs, the probability of a first return within holding time $t_h$ is
\begin{equation}
 p_{\rm return}=1-\exp[-\nu k_c e^{-\beta\Delta}t_h].
\end{equation}
For $\eta=0.2$, $\nu=0.01\,\mathrm{s}^{-1}$, $t_h=1000\,\mathrm{s}$, and $\beta\Delta=12$, the worst value is about $0.0030572$. These parameters define an illustrative mesoscopic model. Finite barriers imply finite retention times, and reversing the transition rates additionally requires a specification of the reversed source, controller, archive, and bath preparation.

For the fixed-past bound, antipodal pairing makes the expectation of an even observable under \cref{eq:conditional} equal to its expectation under $\rho_{xy}$. Thus the expectations of $\cos2\lambda$ in the two $x=0,g=0.2$ target densities differ by $\sech0.4$. For $|f|\leq1$, $|\E_pf-\E_qf|\leq2\TV(p,q)$, with $\TV(p,q)=\tfrac12\int|p-q|$. The triangle inequality yields $\max_y\TV(q,\mu_{0,+,y})\geq\sech0.4/4$. Direct integration gives the stronger numerical pair-distance bound stated in \cref{sec:controller}.

A separate communication constraint applies to a restricted scalar architecture. Let a deterministic message $s=s(x)$ be the only information about $x$ reaching a later correction $m(s,y)>0$. Exact cancellation requires
\begin{equation}
 m(s(x),y)\widehat Z_{xy}=\kappa_x\quad\text{for every }y.
\end{equation}
For axes $\{0,\pi/4,\pi/2,3\pi/4\}$, the rows of $\widehat Z$ are cyclic permutations of $(Z_0,Z_{45},Z_{90},Z_{45})$ with $Z_0>Z_{45}>Z_{90}>0$. Equal messages require proportional rows. Their equal row sums turn proportionality into equality, but their unique maxima lie in different columns. Four messages, or two fixed-length bits, are therefore necessary. Stochastic encoding, shared prior information, outcome-dependent corrections, and other couplings are outside the assumptions of this bound.

\section{Quantum dilations, calibration, and continuity}\label{app:quantum}
An explicit detector circuit uses an input qubit $S$ and blank qubits $E,N,M$, with computational bit zero denoting outcome $+1$. In order, apply $R_y(-x)$ to $S$; swap $S,E$; apply $R_y(\vartheta)$ to $N$, with $\vartheta=2\arcsin\sqrt{(1-v_\gamma)/2}$; apply controlled-NOT gates $E\to S$, $N\to S$, and $S\to M$; then apply $R_y(x)$ to $S$. Here $R_y(\theta)=\exp(-i\theta Y/2)$. For rotated input $\sum_{j=0}^1 c_j\ket j$, $E$ retains $j$ and $N$ has amplitudes $\sqrt{1-p},\sqrt p$, where $p=(1-v_\gamma)/2$. The first two controlled-NOTs set $S=j\oplus n$, and the third copies that bit to $M$. Measuring $M$ and tracing $E,N$ removes different-$j,n$ coherences and yields \cref{eq:instrument}. All gates are unitary; retaining their environments is necessary to invert the circuit.

The Werner source admits the purification
\begin{equation}
 \ket{\Psi_w}=\sum_{j=0}^3\sqrt{p_j}\ket{B_j}_{AB}\ket j_C,\qquad
 p_0=(1+3w)/4,\quad p_1=p_2=p_3=(1-w)/4,
\end{equation}
where $\ket{B_0}=\ket{\psi^-}$ and the remaining states are the other three orthonormal Bell states. Completing this normalized vector to an orthonormal basis supplies an ideal unitary preparation, not a material Hamiltonian or a derivation of retrocausal hidden variables.

The hybrid matrix is explicitly
\begin{equation}
 \Omega_{xy}=\frac{\id_4}4-\frac{1-\Re h_{xy}}8X\otimes X
 -\frac{1+\Re h_{xy}}8Z\otimes Z
 -\frac{\Im h_{xy}}8(X\otimes Z+Z\otimes X).\label{eq:hybridmatrix}
\end{equation}
The first harmonics vanish by antipodality; $\E\sin^2\lambda=(1-\Re h_{xy})/2$, $\E\cos^2\lambda=(1+\Re h_{xy})/2$, and $\E\sin\lambda\cos\lambda=\Im h_{xy}/2$ prove the formula. At the two tested contexts, mixed terms vanish and both remaining products commute with $A=Z\otimes Z$.

For the first weak meter, prepare a blank qubit $M$ in $\ket0$, put $\theta=\arcsin k$, apply $U_\theta=\exp[-i\theta A\otimes Y_M/2]$, and measure $M$ in the $X$ basis with result $r$. Since $A^2=\id_4$,
\begin{equation}
 \bra{r_X}U_\theta\ket0=\frac{\cos(\theta/2)\id_4+r\sin(\theta/2)A}{\sqrt2}=M_r^{(k)}.
\end{equation}
For $0\leq\theta\leq\pi/2$ the operator is positive and its square is $F_r^{(k)}$. The ideal pulse Hamiltonian is $\hbar\dot\theta(t)A\otimes Y_M/2$. Alternatively prepare $M$ in $\cos\alpha\ket0+\sin\alpha\ket1$, where $\alpha=\tfrac12\arccos k$, apply a controlled-NOT from each source qubit into $M$, and measure $M$ computationally. Even and odd parity amplitudes give the same two Kraus operators, with bit zero corresponding to $r=+1$. Neither circuit is a Hamiltonian for the augmented retrocausal fields. The nonselective channel is
\begin{equation}
 \mathcal N_k(\Omega)=\cos^2(\theta/2)\Omega+\sin^2(\theta/2)A\Omega A,
\end{equation}
which leaves the two hybrid states unchanged by \cref{eq:hybridmatrix}.

The calibration-error correction follows from $|\Tr[(G_{+,j}-F_+^{(k)})\sigma_j]|\leq\eta$ for each context. Their probability difference changes by at most $2\eta$, giving the bound in \cref{sec:quantum}. The same trace-distance argument has a classical version for a fixed response $L_k(r\mid\lambda)=[1+rkf(\lambda)]/2$, $|f|\leq1$, with $D$ replaced by total variation. For $f=\cos2\lambda$, the two baseline means differ by $q$, so exact protection requires $\max_j\TV(q_j(k),\rho_j)\geq q/4$. Sharp attainability for those angle densities is not claimed. A quantum spin effect need not implement this classical harmonic response.

For the smooth counterexample, direct subtraction of \cref{eq:smoothresponse} gives
\begin{equation}
 \TV(\widetilde\rho_u,\rho_u)=\frac{k^2|m_u|}{2(1-k^2m_u^2)}\E_{\rho_u}|m_u-f|,
\end{equation}
proving the stated $O(k^2)$ convergence at fixed $u$. A sufficient continuity criterion is the following: on a finite configuration space, nonnegative weights $w_k(\zeta)$ continuous at $k=0$ and $Z_0=\sum_\zeta w_0(\zeta)>0$ imply total-variation convergence of $w_k/Z_k$ to $w_0/Z_0$. In a continuous space, it suffices that $w_k\to w_0$ almost everywhere and $w_k\leq G$ near zero for an integrable $G$, with $Z_0>0$. Dominated convergence gives $L^1$ convergence of weights and $Z_k\to Z_0$; normalization then gives the result. If a fixed map $\tau(\zeta)$ assigns a unit-trace positive matrix to each hidden configuration, then
\begin{equation}
 \sigma_k=\int\tau(\zeta)\mu_k(\dd\zeta),\qquad D(\sigma_k,\sigma_0)\leq\TV(\mu_k,\mu_0).
\end{equation}
This follows by integrating $\|\tau(\zeta)\|_1=1$ against the absolute difference measure. Such regularity, a fixed hybrid state map, and the calibrated meter cannot jointly give exact record protection while continuously recovering the two original states. Singular limits can violate regularity; contextual responses can violate calibration. Locality or time-reversal symmetry alone establishes neither premise.

\section{Environmental, clock, archive, and reset partners}\label{app:partners}
Retain the local isometric environment in \cref{eq:dilation}. A measurement of that environment gives a completely positive refinement $\mathcal J_{a,e}$ of the coarse instrument. Choose a strictly positive normalized raw reference $q(e\mid a)$ and take
\begin{equation}
 W_{a,e}=q(e\mid a)W_a,\qquad\sum_e q(e\mid a)=1.
\end{equation}
At fixed coarse context and unchanged continuation, $R_{h,e}=q(e\mid a)R_h$, $d_{h,e}=q(e\mid a)d_h$, and therefore
\begin{equation}
 \sum_e\frac{n_{h,e}}{d_{h,e}}R_{h,e}=\frac{\sum_e n_{h,e}}{d_h}R_h.
\end{equation}
Coarsening restores the hidden branch only when refined quantum maps sum to the same coarse map and the stated proportionality of raw measures holds. A policy changing subsequent apparatus based on $e$ defines a different experiment; its earlier operational record is still protected by \cref{thm:embedding}.

The numerical refinement mixes the two operators $K_{a,+},K_{a,-}$ by the real unitary rotation
\begin{equation}
 \widetilde K_{a,e}=\sum_s U_{es}K_{a,s},\qquad
 U=\begin{pmatrix}\cos\vartheta&-\sin\vartheta\\\sin\vartheta&\cos\vartheta\end{pmatrix},\qquad\vartheta=0.413.
\end{equation}
The refined maps $\widetilde K_{a,e}(\cdot)\widetilde K_{a,e}^\dagger$ sum to the original instrument; $q(e\mid a)=1/2$ supplies support. This rotation is one explicit finite test. The theorem covers other admitted local environmental instruments when their raw partners satisfy its assumptions.

A normalized clock can multiply both $W$ and $\mathcal J$ by the same scalar density $f(t)=\nu e^{-\nu t}$. That scalar cancels in $C$, and the complete clock-outcome record has the paired quantum law. A finite deadline requires a no-event branch to keep the instrument complete. A hidden-angle-dependent raw clock requires a paired operational instrument satisfying the support condition; an averaged rate alone does not specify such an instrument.

For orthogonal archive, display, and blank registers $H,D,B$, a controlled copy followed by a display--blank swap acts on the ready subspace as
\begin{equation}
 \ket{0,r,0}_{HDB}\longmapsto\ket{r,0,r}_{HDB}.
\end{equation}
Use the same unitary permutation as the raw register map. The display resets, the archive remains, and the environment keeps the displaced value. Further cycles require additional blank registers or a specified erasure operation. More generally, a normalized raw reset and trace-preserving quantum reset satisfy
\begin{equation}
 W_{\rm reset}1=1,\qquad\mathcal J_{\rm reset}^*(\id)=\id.
\end{equation}
At unused terminal outputs they leave incoming fields unchanged. A measured reset environment is retained as a refined instrument. Inversion requires actual correlated outputs, not fresh blanks. These statements specify the register layer, not a reversed source-field dynamics or material bath for the potential in \cref{eq:factor}.

\bibliographystyle{quantum}
\bibliography{references}
\end{document}